\documentclass[12pt]{article}
\usepackage{amssymb}
\usepackage{amsmath,mathtools}
\usepackage{theorem}
\usepackage{verbatim,xcolor}
\usepackage{natbib}

\theorembodyfont{\rmfamily}

\newtheorem{corollary}{Corollary}

\newtheorem{example}{Example}

\newtheorem{lemma}{Lemma}

\newtheorem{proposition}{Proposition}
\newtheorem{remark}{Remark}

\newenvironment{proof}[1][Proof]{\textbf{#1.} }{\ \rule{0.5em}{0.5em}}
\newcommand{\df}[1]{\textit{#1}}

\def\Re{\mathbb R}
\def\NE{\mathrm{NE}}
\def\U{\mathcal{U}}
\def\mbf{\mathbf}

\newcommand{\wtd}[1]{\widetilde{#1}}

\begin{document}

\title{Manipulation of Belief Aggregation Rules}
\author{Christopher P. Chambers, Federico Echenique, Takashi Hayashi}
\date{\today}
\maketitle

\begin{abstract}
This paper studies manipulation of belief aggregation rules in the setting where the society first collects individual's probabilistic opinions and then solves a public portfolio choice problem with common utility based on the aggregate belief.

First, we show that belief reporting in Nash equilibrium under the linear opinion pool and log utility is identified as the profile of state-contingent wealth shares in parimutuel equilibrium with risk-neutral preference.

Then we characterize belief aggregation rules which are Nash-implementable. We provide a necessary and essentially sufficient condition for implementability, which is independent of the common risk attitude.
\end{abstract}

\section{Introduction}
When people disagree on prior beliefs about uncertain states of the world in a stubborn manner, as is common in environmental issues and health issues, the society needs to aggregate those beliefs in order to make a pubic decision.\footnote{By ``stubborn,''  we mean we are working outside the framework after  \citet{harsanyi1967games,harsanyi1968games} and \cite{aumann76} that differences in beliefs should come only from difference in information received.}

Such a problem will be set in the following timeline: the society first collects individuals' probabilistic opinions and aggregates them with certain rule, and then makes a public decision based on the aggregate belief.

A natural problem arising there is that each individual may have an incentive to misreport his/her belief, particularly by exaggerating it, so as to manipulate the aggregate belief and resulting public decision in favor of his/her own belief. As a result, we observe an apparent polarization in reported beliefs. Such problem seems quite pervasive in environmental issues and health issues.

\bigskip

To illustrate, imagine that the society aggregates individuals' reported beliefs with the linear opinion pool and then solves a public portfolio choice problem in which the individuals have a common taste and common risk attitude.
\begin{example}
There are two possible states, 1 and 2. There are two individuals, $a$ and $b$. Let $p_{i}$ denote individual $i$'s 
truly believed subjective probability of state 1, where $i=a,b$.
The individuals have a common interest in collective wealth and have common risk attitude described by vNM index $u(z)=\ln z$.  This index is adopted in the collective decision as well. 

Suppose the group aggregates their subjective probabilities according to the linear opinion pool: 
\[f(\rho_{a},\rho_{b})=\frac{\rho_{a}+\rho_{b}}{2}.\]
Here, this equation specifies the aggregate probability of state 1, where $\rho_{i}$ denotes $i$'s reported probability of state 1, which may or may not be equal to $p_{i}$.

Society solves a portfolio choice problem:
\[\max_{x}\frac{\rho_{a}+\rho_{b}}{2}\ln x+
\left(1-\frac{\rho_{a}+\rho_{b}}{2}\right)\ln (1-x),\]
where we assume for simplicity that the initial collective wealth is 1 unit and that state-prices of the two states are the same.
Then the optimal portfolio choice given $(\rho_{a},\rho_{b})$ is 
simply 
\[x(\rho_{a},\rho_{b})=\frac{\rho_{a}+\rho_{b}}{2}.\]

Thus, each individual's payoff function in the game of belief reporting is
\[U_{i}(\rho_{a},\rho_{b})=p_{i}\ln \left(\frac{\rho_{a}+\rho_{b}}{2}\right)+ (1-p_{i})\ln \left(1-\frac{\rho_{a}+\rho_{b}}{2}\right), \ \ \ i=a,b\]

Then $a$'s best 
response is
\[BR_{a}(\rho_{b})=
\left\{\begin{array}{ll}
1 & \text{if} \ \ 0 \le \rho_{b} \le 2p_{a}-1\\
2p_{a}-\rho_{b} & \text{if} \ \  
2p_{a}-1 \le \rho_{b} \le 2p_{a} \\
0 & \text{if} \ \ 2p_{a} \le \rho_{b} \le 1
\end{array}
\right.
\]
where the first and the third cases are mutually exclusive. Likewise, $b$'s best 
response is
\[BR_{b}(\rho_{a})=
\left\{\begin{array}{ll}
1 & \text{if} \ \ 0 \le \rho_{a} \le 2p_{b}-1\\
2p_{b}-\rho_{a} & \text{if} \ \ 
2p_{b}-1 \le \rho_{a} \le 2p_{b} \\
0 & \text{if} \ \ 2p_{b} \le \rho_{a} \le 1
\end{array}
\right.
\]

Without loss, assume $p_{a}\le p_{b}$. Then pure-strategy Nash equilibria $(\rho^{\ast}_{a},\rho^{\ast}_{b})$ are given by:
\[(\rho^{\ast}_{a},\rho^{\ast}_{b})=
\left\{\begin{array}{ll}
(0,2p_{b}) & \text{if} \ \  p_{a}<p_{b}\le \frac{1}{2}\\
(0,1) & \text{if} \ \ p_{a}\le \frac{1}{2}\le p_{b} \ \ \text{and}
\ \ p_{a} <p_{b}\\
(2p_{a}-1,1) &  \text{if} \ \ \frac{1}{2} \le p_{a}<p_{b} \\
\text{any} \ \ \rho_{a}+\rho_{b}=2p &\text{if} \ \ p_{a}=p_{b}=p.
\end{array}
\right.
\]
The resulting aggregate belief is independent of the actual Nash equilibrium, and is uniquely given by
\[f^{\ast}(p_{a},p_{b})=
\left\{\begin{array}{ll}
p_{b} & \text{if} \ \  p_{a}<p_{b}\le \frac{1}{2}\\
\frac{1}{2} & \text{if} \ \ p_{a}\le \frac{1}{2}\le p_{b} \ \ \text{and}
\ \ p_{a} <p_{b}\\
p_{a} &  \text{if} \ \ \frac{1}{2} \le p_{a}<p_{b} \\
p &\text{if} \ \ p_{a}=p_{b}=p
\end{array}
\right.
\]
which is the median of $\{p_{a},p_{b},\frac{1}{2}\}$, counting possible replications.

\end{example}

Thus, in equilibrium we typically observe a polarized distribution of \emph{reported} beliefs regardless of the distribution of truly believed beliefs, and the resulting aggregate belief is the median of three beliefs, where 1/2 is a ``phantom'' voter.

\bigskip

Motivated by the problem illustrated with the previous example, we propose a model of public decision making and belief aggregation. We assume that society first aggregates individuals' probabilistic beliefs according to a certain rule, an opinion pool, and then solves a public portfolio choice problem in which everybody has a common interest. All individual agents wish to increase social wealth, and have the same risk attitude. 

The portfolio choice setting is an abstraction meant to capture the problem of choosing a public policy in the face of uncertainty, when agents' differ about their prior beliefs. The problem can involve, for example, environmental or health policies. Our assumption that agents share tastes and risk attitudes is done for the purpose of isolating, and focusing on, the issue of belief aggregation.

We assume that individual agents disagree about their prior beliefs, and that  there is  \textit{complete information} about it among them. The social planner, on the other hand, is ignorant about agents' beliefs.  In other words, 
agents knows who is exaggerating their beliefs while the social planner does not.

We tackle two questions. One is a positive question regarding the social outcome when individual agents play the game of reporting their beliefs in a non-cooperative manner.  As suggested by the example, it is expected that each individual exaggerates their belief in order to manipulate the aggregation rule and the resulting aggregate belief. When there are more than two states, and more than two individuals, however, it is not clear what is meant by ``exaggerating'' one's belief. For the general case of many states and many individuals, we show that belief reporting in Nash equilibrium under the linear opinion pool and log utility is identified as the profile of state-contingent wealth shares in parimutuel equilibrium with risk-neutral preference.

The second question we tackle is what the class of belief aggregation rules that are implementable looks like. We provide a necessary condition for implementability, and show that it is also sufficient under a mild unanimity condition. Remarkably, we show that the condition for implementability is equivalent to a tractable condition that is independent of the agents' common risk attitude.

\subsection*{Related Literature}
To our knowledge, there are two types of existing studies on manipulation of aggregation under stubborn disagreements in subjective beliefs.

One is to propose a `reduced-form' definition that an agent wants to manipulate the aggregate belief to be `closer' to his/her belief, which does not specify 
the public decision problem to be solved after the aggregation
or the individuals' underling preferences, where the modeller imposes an exogenous measure of `closeness' (\cite{varloot2022level} and \cite{dietrich2023impossibility}).

We chose to explicitly model a particular public decision problem and the agents' underlying preferences over outcomes, as 
we view that we cannot make a prediction of actual manipulation behaviors without specifying them. 
In order to isolate the issue of manipulating belief aggregation rules from the rest, we consider a setting that the individuals have a common interest in the economic outcomes and a common risk attitude. 

The other is to allow heterogeneity and manipulation in tastes over outcomes as well as in beliefs \citep{bahel2020strategyproof}. They consider a social choice function as a mapping from the set of profiles of subjective expected utility preferences over \textit{all} Savage acts into the set of \textit{all} Savage acts (\cite{savage72}, \cite{anscombeaumann63}). They show that it satisfies strategy-proofness and unanimity only if it is a state-by-state dictatorship. Their result relies on the product-set structure of the set of all Savage acts and separable preferences, while they could allow state-dependent feasible sets and state-dependent outcome preferences, which rules out dealing with transferring resource from one state to another. It should be noted that the set of portfolio choices (i.e., the budget set) in our setting is not a product set, and even expected utility preferences does not allow defining separability when they are restricted to such non-product set.

\bigskip

Apart from the problem of manipulation, there is a huge literature on aggregating beliefs/priors and aggregating preferences under disagreements in beliefs in general. The literature in statistics, where an aggregation rule is called an opinion pool, is comprehensively surveyed by \cite{genestzidek86}. Among various alternative opinion pools, the linear opinion pool \cite{stone61} has received substantial attention, and was axiomatically characterized by \cite{mcconway81}. 

The parimutuel method of belief aggregation was proposed by \cite{eisenberggale59}. The parimutuel method is a market-like mechanism, in which the aggregate probability incorporates individual beliefs through the form of market betting odds. 

When people disagree on both beliefs and tastes, it is known that  only dictatorship can satisfy both the Pareto principle applied to uncertain prospects and 
subjective expected utility theory \citep{mongin1995consistent}. This is another reason why we restrict attention to the case where individuals have the same taste over outcomes and the same risk attitude.





\section{Belief Aggregation and Public Portfolio Choice}\label{sec:model}
Let $S$ be a finite set of states of the world. A \df{belief} is a probability distribution over $S$. Let  $\Delta(S)$ denote the set of all probability distributions over $S$. Let $I=\{1,\cdots,n\}$ be the set of individuals, who may differ in their beliefs. 

A \df{belief aggregation rule} is a mapping $f: \Delta(S)^{I} \rightarrow \Delta(S)$.

A \df{social outcome} is an element of $\mathbb{R}_{+}$, and a state-contingent outcome, or \df{portfolio}, is a function $x:S\to\Re_+$. 

We may interpret outcomes as social wealth, and assume that all individuals have a common interest in social wealth being larger rather than smaller.
We assume that individuals share the same risk attitudes as well, which are described by a von-Neumann/Morgenstern index 
$u: \mathbb{R}_{+} \rightarrow \mathbb{R}$. We assume that $u$ is  twice-continuously differentiable on $\mathbb{R}_{++}$ and satisfies $u^{\prime}>0$ and $u^{\prime\prime}<0$, $\lim_{z\rightarrow 0}u^{\prime}(z)=\infty$ and $\lim_{z \rightarrow \infty}u^{\prime}(z)=0$. Given a portfolio $x$, the function $s\mapsto u(x_s)$ is a \df{utility act}.
 
Society solves a portfolio choice problem, given an aggregate belief. For simplicity, we focus on the case of a linear budget set, or set of feasible state-contingent outcomes:
\[X=\{x \in \mathbb{R}^{S}_{+}: \pi \cdot x = w\}.\]
The set $X$ is the set possible portfolio choices, where $\pi \in \mathbb{R}^{S}_{++}$ is a price vector for Arrow securities, and $w>0$ is the initial social wealth. 
What is essential here is that the dimension of $X$ is $|S|-1$.


Given an aggregate belief $p \in \Delta(S)$, let 
\[x(p)=\arg\max_{x \in X} \sum_{s\in S} u(x_{s})p_{s}\]
be the socially optimal portfolio, and let 
\[u(p)=(u(x_{s}(p)))_{s \in S}\]
be the optimal utility act obtained when society adopts $p$ as its aggregate belief. It is worth noting that the objective function $\sum_{s\in S} u(x_{s})p_{s}$ is additively separable over $\mathbb{R}^{S}_{+}$, but \textit{not over} $X$.

\bigskip

We can describe the collective decision as a choice over utility acts. 
Let $\mathcal{U}=\{(u(x_{s}))_{s \in S} \in \mathbb{R}^{S}: x \in X\}$ be the set of possible utility acts induced by portfolio choices. For each $p \in \Delta(S)$, 
define
\[u(p)=(u(x_{s}(p)))_{s \in S} .\]
to be the utility act that is chosen as optimal when the society adopts belief $p$.

The claim below, which is standard, states that there is one-to-one correspondence between adopted belief and resulting portfolio choice, and also between adopted belief and resulting utility act.
\begin{proposition}
For all $x \in X$ there is unique $p \in \Delta(S)$ such that $x(p)=x$.

For all $u \in \mathcal{U}$ there is unique $p \in \Delta(S)$ such that $u(p)=u$.
\end{proposition}

Under the belief aggregation rule $f$, when a profile of beliefs $\mathbf{p} \in \Delta(S)^{I} $ is reported by the individual agents (truthfully or not), society chooses a portfolio $x(f(\mathbf{p}))$. This portfolio choice results in the utility act $u(f(\mathbf{p}))\in \U$.

\section{The Linear Opinion Pool Preference Revelation Game and a Parimutuel Market}\label{sec:nasheq}

The \df{symmetric linear opinion pool} is the belief aggregation rule defined by $f(\mathbf{p}) = \frac{1}{n}\sum_{i=1}^n p_i$.    Assume that $u(z)= \ln(z)$, and that $\pi=(1,1,\ldots,1)$, with $w=1$; therefore $X$ is itself a ``probability simplex.''

Under these assumptions, the following is well-known:

\begin{proposition}
For all $p\in \Delta(S)$, $x(p)=\arg\max_{x\in X} \sum_{s\in S}p_{s}\ln(x_{s})=\{p\}$.
\end{proposition}

More generally, for arbitrary $\pi$ or $w$, the solution gives $\frac{wp_{s}}{\pi_{s}}$ in each state, that is, the barycenter of $X$ with coordinates $p$.

Define the \df{log utility preference revelation game} for profile $(p_1,\ldots,p_n)$ to be the game $(\Delta(S),U_{p_i})_{i\in N}$ in which each player $i\in N$ has $\Delta(S)$ as their strategy space, and payoff function 
\[
U_{p_i}(\rho_1,\ldots,\rho_n) = \sum_{s\in S} p_{i,s } \log\left(\frac{1}{n}\sum_{j=1}^n \rho_{j,s}\right).
\]
Observe that the payoff function is well defined under the convention that $0\times \log (0)=0$: if, for some $s\in S$,  $\sum_j p_{j,s} = 0$, then $p_{j,s}=0$ and  $\log(0)$ does not enter into the sum for any individual~$j$.  Technically speaking, this utility function can take the value $-\infty$, which will not occur at an optimum.

To understand the Nash equilibria of the belief revelation game $(\Delta(S),U_{p_i})_{i\in N}$ we need to introduce an auxiliary betting market. The model is borrowed from \cite{eisenberggale59}, but it boils down to the classical notion of Walrasian equilibrium for linear preferences and nominal wealth, as opposed to endowments.

A \df{parimutuel market} is a collection $(\Re^S_+,p_i)_{i\in N}$ that describes a collection of consumers $i\in N=\{1,\ldots,n \}$. Each consumer has $\Re^S_+$ as consumption space, and a linear utility function $x\mapsto p_i\cdot x$. We interpret $x\in\Re^S_+$ as state-contingent consumption, and the linearity of utility $p_i\cdot x$ as risk neutrality. 

A \df{parimutuel equilibrium with equal wealth} is a tuple $(\rho, \mathbf{x})$ in which $\rho\in\Delta(S)$ is a \df{price vector}, and $\mathbf{x} = (x_1,\ldots,x_n)\in (\Re^S_+)^N$ satisfies:
\begin{enumerate}
\item\label{it:opt} If $\rho \cdot y \leq 1/n$, then $p_i \cdot x_i \geq p_i \cdot y$.
\item\label{it:supdem} $\sum_{i=1}^n x_i = (1,\ldots,1)$.
\end{enumerate}
The latter means that $\mathbf{x}$ is an \df{allocation} of the aggregate quantity $(1,\ldots,1)$. Condition~\eqref{it:opt} means that agent $i$ optimizes their linear utility function by choosing $x_i$ in a budget set defined by prices $\rho$ and ``nominal wealth'' $1/n$.\footnote{Observe that we could replace the ``nominal wealth'' hypothesis with the hypothesis that each individual is endowed $(1/n,1/n,\ldots,1/n)$, that is; the parimutuel equilibrium with equal wealth is the same as a Walrasian equilibrium from equal division.} Condition~\eqref{it:supdem} means that markets clear for the consumption in every state.

The wealth constraint of $1/n$ is what allows us to conclude that $\rho$ is a probability measure.\footnote{If we used the Walrasian model with endowments, then the price would have to be renormalized in general in order to ensure its coordinates sum to one.}

\begin{proposition}\label{prop:EG}
A profile $(\rho_1,\ldots,\rho_n)$ is a Nash equilibrium for the log utility preference revelation game $(\Delta(S),U_{p_i})_{i\in N}$ if and only if there is a parimutuel equilibrium with equal wealth $(\rho, \mathbf{x})$ of the parimutuel market $(\Re^S_+,p_i)_{i\in N}$ with: \[
\rho_{is} = \frac{\rho_{s}x_{is}}{\sum_{s^{\prime}}\rho_{s^{\prime}}x_{is^{\prime}}}
\] for all $i\in N$.  \\
In particular $f(\rho_1,\ldots,\rho_n)=\rho$.
\end{proposition}

In the parimutuel betting model, $\rho_{s}x_{is}$ is the amount of money that individual $i$ bets on state $s$. By Proposition~\ref{prop:EG}, the Nash equilibrium report $\rho_i$ equals, in each state, the fraction that $i$ bets on the state over the total bets on the state: these coincide with the fraction of the total betting pool that goes to agent $i$ if the state is realized. For example, in the race-track betting story of \cite{eisenberggale59}, each state is a different horse that may win a race. If horse $s$ wins and the total betting pool (after the track's profits are deducted) is $M$, then $i$ receives a payout that equals $\frac{\rho_{s}x_{is}}{\sum_{s^{\prime}}\rho_{s^{\prime}}x_{is^{\prime}}} \times M$.

\begin{proof} 
The idea is to write equilibrium optimality conditions as first-order conditions for both parimutuel and Nash and observe they are essentially the same. We assume that every individual's belief is full-support, as the argument extends to boundaries in a continuous manner.

The first order conditions for the parimutuel problem in equilibrium are:
For all $i\in I$, there is $\lambda_i>0$ for which
\begin{enumerate}
\item If $x_{is}>0$, then $p_{is} = \lambda_i \rho_{s}$
\item If $x_{is} = 0$, then $p_{is}\leq \lambda_i \rho_{s}$.
\end{enumerate}
This, together with the hypothesis that for all $i\in I$, $\sum_{s} \rho_{s}x_{is}=\frac{1}{n}$ are what constitute a parimutuel equilibrium.

The first order conditions for the Nash revelation game in equilibrium are:
For all $i\in I$, there is $\lambda_i > 0$ for which
\begin{enumerate}
\item If $\rho_{is}>0$, then $p_{is} = \lambda_i \frac{1}{n}\sum_{j} \rho_{js}$
\item If $\rho_{is} = 0$, then $p_{is} \leq \lambda_i \frac{1}{n}\sum_{j} \rho_{js}$.
\end{enumerate}

Now, suppose that $(\rho,x)$ constitutes a parimutuel equilibrium.  Then we know that $\sum_s \rho_{s}x_{is} = \frac{1}{n}$ (this is the budget constraint and Walras' law).  So for each agent, let $\rho_{is} = n\rho_{s}x_{is}$ and observe this is a probability distribution for each agent.  Further:  $\frac{1}{n}\sum_{i} \rho_{is}=\frac{1}{n}\sum_{i} n\rho_{s}x_{is}=\rho_{s}\sum_i x_{is} = \rho_{s}$.  

Observe further that $x_i(s) > 0$ iff $\rho_i(s) > 0$, so that the first order conditions for the Nash revelation game are now satisfied (as $\rho(s)=\frac{1}{n}\sum_j \rho_j(s)$).  

In the other direction, suppose that $(\rho_1,\ldots,\rho_n)$ constitutes a Nash equilibrium of the preference revelation game.  Then let $\rho=\frac{1}{n}\sum_i \rho_i$ and let $x_{is} = \frac{\rho_{is}}{\sum_j \rho_{js}}$.  
Observe that $\sum_i x_{is} = 1$ so this is indeed an allocation of $(1,1,\ldots,1)$.  Further observe that $\sum_s \rho_{s}x_{is} =\frac{1}{n} \sum_{s} \sum_{j} \rho_{js} \frac{\rho_{is}}{\sum_{j} \rho_{js}}=\frac{1}{n}\sum_{s}  \rho_{is}=\frac{1}{n}$.

Finally, again $x_{is} > 0$ iff $\rho_{is} > 0$, so the FOC's for the parimutuel problem are satisfied.  \end{proof}

\subsection{Two states}

Here we characterize the parimutuel prices that arise for two states of the world.  The Nash equilibria will follow simply from the possible equilibrium allocations.  Let us call the states of the world $1$ and $2$, where $S = \{1,2\}$.  Any individual $i\in N$ is characterized by the probability they assign to $s_1$.  So, following the example, we represent $p\in \Delta(S)$ by the probability $p$ assigns to state $1$.

We show that probability aggregation resulting in equilibrium is a specific version of the generalized median rule by \cite{moulin1980strategy}, in which equally-spaced probability values play the role of phantom voters.  In the Proposition, the notation $\mbox{med}(x_1,\ldots,x_m)$ when $m$ is odd refers to the median value of the vector $(x_1,\ldots,x_m)$:  that is, the value $x^*$ for which $|\{i:x_i < x^*\}|\leq \frac{m-1}{2}$ and $|\{i:x_i > x^*\}|\leq \frac{m-1}{2}$.  
\begin{proposition}\label{prop:twostates}Given $p_1,\ldots,p_n\in [0,1]$, the parimutuel prices are given by $$\mbox{med}\left(p_1,\ldots,p_n,\frac{1}{n},\ldots,\frac{n-1}{n}\right).$$\end{proposition}

\begin{proof} Suppose the median is given by $\frac{m}{n}$.  This means that we can partition $N$ into two groups $U$ and $L$, where if $j\in U$, then $p_j \geq \frac{m}{n}$ and if $j\in L$, then $p_j \leq \frac{m}{n}$ where $|U|=m$ and $|L|=n-m$.  Consequently by setting the price $\rho=\frac{m}{n}$, for all $j\in U$, setting $x_{j1} = \frac{1}{m}$ and $x_{j2}=0$ is clearly optimal subject to the constraint 
\[\frac{m}{n}x_{j1}+\frac{n-m}{n}x_{j2}\leq \frac{1}{n}.\] Similarly, for all $j \in L$, setting $x_{j1}=0$ and $x_{j2}=\frac{1}{n-m}$ is optimal.  Observe that these allocations satisfy the market clearing conditions so that $\frac{m}{n}$ is then the equilibrium price (there may be potentially many equilibrium allocations depending on how many agents have $p_j = \frac{m}{n}$).

Now suppose the median is given by $p_i$ for some $i\in N$.  Then there are several possibilities:  either $p_i < \frac{1}{n}$, $p_i > \frac{n-1}{n}$, or $\frac{m}{n}<p_i<\frac{m+1}{n}$ for some $m$  (if $p_i = \frac{m}{n}$ for some $m$, we are back in the preceding case).

In the first case, this means that for all $j$, $p_j \leq p_i$.  Now set $\rho = p_i$, and observe that for all $j\neq i$, $x_{j1} = 0$ and $x_{j2} = \frac{1}{n(1-p_i)}$ is an optimal demand subject to the budget constraint; further, observe that $\frac{n-1}{n(1-p_i)}<1$ as $1-p_i > \frac{n-1}{n}$, so that the consumption of agents $j\neq i$ does not exhaust the market clearing constraint.  Therefore agent $i$ needs to consume $x_{i1}=1$ and $x_{i2}=1-\frac{n-1}{n(1-p_i)}$.  Observe that agent $i$ is indifferent between all allocations exhausting her budget constraint and that this allocation indeed exhausts the budget constraint.

The second case is symmetric.

Let us now handle the third case and fix $i\in N$ for which $p_i$ which is the median, where $\frac{m}{n}<p_i<\frac{m+1}{n}$.  Then we may partition $N\setminus\{i\}$ into two groups $U$ and $L$ where if $j\in U$, then $p_j \geq p_i$ and if $j\in L$, then $p_j \leq p_i$, where $|U|=m$ and $|L|=n-1-m$.  For each $j\in U$, setting $x_{j1}=\frac{1}{np_i}$ and $x_{j2}=0$ maximizes the preference subject to the budget constraint:  observe that $\frac{m}{np_i}<1$ as $p_i > \frac{m}{n}$, so that the consumption of agents in $U$ does not exhaust the market clearing constraint.  Similarly, for all $j\in L$, we set $x_{j1}=0$ and $x_{j2}=\frac{1}{n(1-p_i)}$ and observe that since $\frac{n-1-m}{n(1-p_i)}<1$ as $p_i<\frac{m+1}{n}$.  Finally, setting consumption of agent $i$ as $x_{i1}=1-\frac{m}{np_i}$ and $x_{i2}=1-\left(\frac{n-1-m}{n(1-p_i)}\right)$ will be optimal so long as this exhausts $i$'s budget constraint, which it does. 
\end{proof}

\begin{remark}Proposition~\ref{prop:twostates} easily demonstrates that the parimutuel price aggregation rule is Pareto-efficient for the case of two states, independently of risk preferences.  However, this fails with three states or more.  Consider for example two individuals with log-preferences, with priors $(1/2,1/2,0)$ and $(0,1/2,1/2)$.  The unique parimutuel price for these preferences is $(1/3,1/3,1/3)$.  However, each individual prefers $(1/4,1/2,1/4)$.  To see this, it is very simple to demonstrate algebraically that $$(1/2)\ln(1/4)+(1/2)\ln(1/2)>(1/2)\ln(1/3)+(1/2)\ln(1/3).$$  Thus, even though parimutuel prices arise from efficient allocation in a linear private goods economy, the efficiency property does not transfer to public good economies.\end{remark}

\section{Nash-implementable Belief Aggregation}\label{sec:implementation}
We turn to Nash implementation of belief aggregation rules, and provide a characterization of the rules that are implementable in Nash equilibrium. Here we move away from the assumption of logarithmic utility used in Section~\ref{sec:nasheq}, and consider instead a general von-Neumann Morgenstern utility index $u$, under the assumptions that we laid out in Section~\ref{sec:implementation}. When a planner has \df{aggregate beliefs} $q$, presumably as the outcome of some aggregation rule that takes individual beliefs as inputs, they choose a utility act $u(q)\in \U$; see Section~\ref{sec:model}.

An individual agent with belief $p_i$ has then a utility $q\mapsto \sum_{s\in S}p_{i,s}u_s(q)$ over aggregate beliefs. We denote the resulting preference over aggregate beliefs by $R(p_{i})$. That is:
\[q \mathrel{R(p_{i})} q^{\prime} \ \ \ \Longleftrightarrow \ \ \ \sum_{s\in S} p_{is} u_{s}(q) \ge 
\sum_{s\in S} p_{is}u_{s}(q^{\prime}).\]

A \df{game form} is a tuple $\Gamma=((M^{i})_{i \in I}, g)$, comprised of a \df{message space} $M^i$ for each individual $i\in I$, and an \df{outcome function} $g: \prod_{i \in I}M^{i}\rightarrow \Delta(S)$. Note that an outcome function chooses an aggregate belief for each message profile chosen by the individual agents.

Given a profile of beliefs $\textbf{p}\in \Delta(S)^I$, a game form $\Gamma$ defines a normal-form game $(M^i,V_{g,p_i})_{i\in I}$ where $M^i$ is $i$'s set of strategies, and $V_{g,p_i}\colon \prod_{j\in I}M^j\to \Re$ is a payoff function for player $i$ that satisfies 
\[
V_{g,p_i}(\textbf{m})\geq V_{g,p_i}(\textbf{m}') \Longleftrightarrow g(\textbf{m}) \mathrel R(p_i) g(\mathbf{m}').
\]

A message profile $\mathbf{m} \in \prod_{i \in I}M^{i}$ is \df{Nash equilibrium} in game $(M^i,V_{g,p_i})_{i\in I}$  if 
\[g(\mathbf{m}) \mathrel{R(p_{i})} g(\widetilde{m}^{i},\mathbf{m}^{-i})\]
holds for all $i \in I$ and $\widetilde{m}^{i} \in M^{i}$. 

Let $\NE(\Gamma,\mathbf{p})$ denote the set of Nash equilibria of $(M^i,V_{g,p_i})_{i\in I}$.

A belief  aggregation rule $f$ is \df{Nash-implementable} if there is a game form $\Gamma$ such that 
\[\{f(\mathbf{p})\}=g(\NE(\Gamma,\mathbf{p}))\]
for all $\mathbf{p} \in \Delta(S)^{I}$. In this case, we say that $\Gamma$ (Nash) implements $f$.

Given $p_{i} \in \Delta(S)$ and $q \in \Delta(S)$, let
\[L(q,p_{i})=\left\{q^{\prime} \in \Delta(S): q \mathrel{R(p_{i})} q^{\prime} \right\}\]

The following property is a natural counterpart to the monotonicity condition of \cite{maskin1999nash} for our model.
\begin{description}
\item \textbf{Monotonicity}: 
For all $\mathbf{p}, \widetilde{\mathbf{p}} \in \Delta(S)^{I}$,  if
\[L(f(\mathbf{p}),\widetilde{p}_{i})  \supset L(f(\mathbf{p}),p_{i}) \]
for all $i \in I$, then 
\[f(\widetilde{\mathbf{p}})=f(\mathbf{p}).\]
\end{description}

We may state the Monotonicity property in a different, but equivalent, way. Define the preference relation over utility acts that is induced by a belief $p_{i}$ as
\[u \mathrel{R^{\mathcal{U}}(p_{i})} u^{\prime} \ \ \ \Longleftrightarrow \ \ \ \sum_{s\in S} u_{s}p_{is} \ge 
\sum_{s\in S} u_{s}^{\prime}p_{is},\]
where $u,u^{\prime} \in \mathcal{U}$. 
Given $p_{i} \in \Delta(S)$ and $u \in \mathcal{U}$, let
\[L^{\mathcal{U}}(u,p_{i})=\left\{u^{\prime} \in \mathcal{U}: u \mathrel{R^{\mathcal{U}}(p_{i})}  
u^{\prime} \right\}\]
Then the monotonicity condition is equivalently stated as 
\begin{description}
\item \textbf{Monotonicity (alternative)}: 
For all $\mathbf{p}, \widetilde{\mathbf{p}} \in \Delta(S)^{I}$,  if
\[L^{\mathcal{U}}(u(f(\mathbf{p})),\widetilde{p}_{i})  \supset L^{\mathcal{U}}(u(f(\mathbf{p})),p_{i}) \]
for all $i \in I$, then 
\[f(\widetilde{\mathbf{p}})=f(\mathbf{p}).\]
\end{description}

\begin{proposition}
If a belief aggregation rule is Nash-implementable, then it is monotonic.
\end{proposition}
\begin{proof}
Suppose that $\Gamma$ is a game-form that implements the belief aggregation rule $f$.

Pick any pair  $\mathbf{p}, \widetilde{\mathbf{p}} \in \Delta(S)^{I}$  which satisfies the presumption of Monotonicity.

Pick $\mathbf{m} \in \NE(\Gamma,\mathbf{p})$. By the definition of Nash implementation, we have that  $g(\mathbf{m})=f(\mathbf{p})$ and that 
\[g(\mathbf{m})\mathrel{R(p_{i})} g(\widetilde{m}^{i},\mathbf{m}^{-i})\]
for all $i \in I$ and $\widetilde{m}^{i} \in M^{i}$. 

The it follows from the choice of $\mathbf{p}$ and $\widetilde{\mathbf{p}}$ that 
\[g(\mathbf{m}) \mathrel{R(\widetilde{p}_{i})} g(\widetilde{m}^{i},\mathbf{m}^{-i}),\]
for all $i \in I$ and $\widetilde{m}^{i} \in M^{i}$. 
Hence $\mathbf{m} \in \NE(\Gamma,\widetilde{\mathbf{p}})$.

By the definition of Nash implementation, we have 
$g(\mathbf{m})=f(\widetilde{\mathbf{p}})$, and hence $f(\widetilde{\mathbf{p}})=f(\mathbf{p})$.
\end{proof}

The next property of a belief aggregation rule concerns invariance when a profile of individual beliefs is move towards the aggregate belief chosen by the rule. 
\begin{description}
\item \textbf{Recursive Invariance}: For all $\mathbf{p} \in \Delta(S)^{I}$ and $(\lambda_{i})_{i \in I} \in [0,1]^{I}$ it holds 
\[
f(((1-\lambda_{i})p_{i}+\lambda_{i}f(\mathbf{p}))_{i \in I})=
f(\mathbf{p}).
\]
\end{description}

We prove that recursive invariance is, under our assumptions, equivalent to monotonicity. This result is remarkable because recursive invariance is independent of the risk attitudes reflected in the index $u$.\footnote{Recursive invariance appears in \citet{brady2015spatial,brady2017spatial} in a similar context, where it is shown to be equivalent to Monotonicity in the context of Euclidean preferences.}

\begin{lemma}\label{lem:maskin}
Monotonicity and Recursive Invariance are equivalent.
\end{lemma}

\begin{proof}
We write $H^{-}(\pi,\alpha)$ for the half-space $\{v \in \mathbb{R}^{S}:\pi\cdot v\leq \alpha \}$.    

The proof proceeds by establishing two claims.

Claim 1:  $L^{\mathcal{U}}(p_i,u(f(\mathbf{p})))\subseteq L^{\mathcal{U}}(q_i,u(f(\mathbf{p})))$ iff $$H^{-}(p_i,p_i\cdot u(f(\mathbf{p}))\cap  H^{-}(f(p),f(p)\cdot u(f(\mathbf{p}))\subseteq H^{-}(q_i,q_i\cdot u(f(\mathbf{p})).$$

Claim 2: If $H^{-}(p_{i},p_{i}\cdot u(f(\mathbf{p}))\cap  H^{-}(f(\mathbf{p}),f(\mathbf{p})\cdot u(f(\mathbf{p})))\subseteq H^{-}(q_{i},q_{i}\cdot u(f(\mathbf{p})))$ then there is $\lambda_{i}\in [0,1]$ so that $q_{i}=\lambda p_{i}+(1-\lambda_{i})f(\mathbf{p})$.

First we prove Claim 1. We may without loss of generality take $u(f(\mathbf{p}))=0$ (translate these sets by $u(f(\mathbf{p}))$).

Suppose that $L^{\mathcal{U}}(p_{i},0)\subseteq L^{\mathcal{U}}(q_{i},0)$ and let $v\in H^{-}(p_{i},0)\cap  H^{-}(f(\mathbf{p}),0)$. Since $\mathcal{U}$ is smooth, and $f(\mathbf{p})$ supports $\mathcal{U}$ at $0$, $f(\mathbf{p})\cdot v\leq 0$ implies that there is $\lambda>0$ so that $\lambda v\in \mathcal{U}$.\footnote{If that were not the case, the only point in the intersection of $\mathcal{U}$ and the line $\{\lambda v:\lambda\in\Re \}$ would be $0$. Then, since $\mathcal{U}$ is convex, we may find a hyperplane separating this line from $\mathcal{U}$  and there would then exist two distinct supporting hyperplanes to $\mathcal{U}$ at zero. Smoothness of $\mathcal{U}$ implies that there is only one.} Then $p_{i}\cdot \lambda v\leq 0$ means that $\lambda v\in L^{\mathcal{U}}(p_{i},0)$. So $\lambda v\in L^{\mathcal{U}}(q_{i},0)$ and thus $q_{i}\cdot v\leq q_{i}\cdot u(f(\mathbf{p}))=0$.

Conversely, suppose that $H^{-}(p_{i},0)\cap  H^{-}(f(\mathbf{p}),0)\subseteq H^{-}(q_{i},0)$ and let $v\in L^{\mathcal{U}}(p_{i},0)$. Then $v\in \mathcal{U}$, which implies that $f(\mathbf{p})\cdot v\leq 0$, as $f(\mathbf{p})$ supports $\mathcal{U}$ at $0$, and $p_{i}\cdot v\leq  0$. Hence $v\in H^{-}(q_{i},0)$, so $v\in L^{\mathcal{U}}(q_{i},0)$.

Now we turn to Claim 2. Again take $u(f(\mathbf{p}))=0$. Suppose that the claim is false. Then by a version of separating hyperplane theorem there exists $v$ so that $q_{i}\cdot v> 0 \geq \pi\cdot v$ for all $\pi$ in the cone generated by $p_{i}$ and $f(\mathbf{p})$. In particular, $v\in H^{-}(p_{i},p_{i}\cdot u(f(\mathbf{p}))\cap  H^{-}(f(\mathbf{p})),f(\mathbf{p})\cdot u(f(\mathbf{p})))$, which is a contradiction.
\end{proof}

\begin{description}
    \item \textbf{No Veto Power}: For all $\mathbf{p} \in \Delta(S)^{I}$, if there are $I^{\prime}\subset I$ with $|I^{\prime}|\ge |I|-1$ and $p \in \Delta(S)$ such that $p_{i}=p$ for all $i \in I^{\prime}$, then $f(\mathbf{p})=p$.
\end{description}

\begin{proposition}\label{prop:Nash}
Assume $|I|\ge 3$. Then a belief aggregation rule is Nash-implementable 
if it satisfies No Veto Power and Monotonicity. 
\end{proposition}
\begin{proof}
Consider the game form $\Gamma$ defined by
    \[M^{i}=\Delta(S)^{I} \times \mathbb{N}\]
    for each $i \in I$. A generic element of $M^i$ is denoted by 
    $m^{i}=(\mathbf{p}^{i},k^{i})$. The outcome function $g$ is defined as follows:

\textbf{Rule 1}: If there is $\widetilde{\mathbf{p}}
\in \Delta(S)^{I} $ such that 
$\mathbf{p}^{j}=\widetilde{\mathbf{p}}$ holds for all 
$j$ then 
$g(\mathbf{m})= f(\mathbf{p})$.

\textbf{Rule 2}: If there is $\widetilde{\mathbf{p}}
\in \Delta(S)^{I} $ such that 
$\mathbf{p}^{j}=\widetilde{\mathbf{p}}$ holds for all 
$j$ but $i$, then 
$g(\mathbf{m})= f(\mathbf{p}^{i})$ if 
$ f(\widetilde{\mathbf{p}}) \mathrel{R(\widetilde{p}_{i})} f(\mathbf{p}^{i})$ and 
$g(\mathbf{m})=f(\widetilde{\mathbf{p}})$ otherwise.

\textbf{Rule 3}: Otherwise, play the integer game or the modulo game, and 
let $g(\mathbf{m})=f(\mathbf{p}^{i})$ for $i \in I$ being the winner of the integer/modulo game.

\bigskip

To show that $\{f(\mathbf{p})\} \subset g(\NE(\Gamma,\mathbf{p}))$, 
consider the message profile $\mathbf{m}$ such that 
$m^{i}=(\mathbf{p},k^{i})$ and, say $k^i=1$, for all $i \in I$. Then $g(\mathbf{m})=f(\mathbf{p})$ holds by construction and there is no profitable deviation, hence 
$\mathbf{m} \in \NE(\Gamma,\mathbf{p})$.

\bigskip
To show that $\{f(\mathbf{p})\} \supset g(\NE(\Gamma,\mathbf{p}))$, 
pick any $\mathbf{m} \in \NE(\Gamma,\mathbf{p})$.

Case 1: Suppose there is $\widetilde{\mathbf{p}}\in \Delta(S)^{I}$ such that $m^i_1=\widetilde{\mathbf{p}}$ for all $i\in I$. To use the monotonicity axiom, consider any $q\in\Delta(S)$ with $f(\widetilde{\mathbf{p}})\mathrel{R(\widetilde{p}_{i})} q$ for all $i\in I$. Since $f$ is surjective, there exists some $\mathbf{p}'$ with $q=f(\mathbf{p}')$. Since  $f(\widetilde{\mathbf{p}})\mathrel{R(\widetilde{p}_{i})} q$, any $i$ could make the rule choose $q$. Given that $\mathbf{m}$ is a Nash equilibrium, it must be that $f(\widetilde{\mathbf{p}})\mathrel{R(p_{i})} q$. By Monotonicity, then, we conclude that $f(\mathbf{p})=f(\widetilde{\mathbf{p}})$.

Case 2: Suppose there is $\widetilde{\mathbf{p}}\in \Delta(S)^{I}$ such that $m^i_1=\wtd{\mathbf{p}}$ for all $i\in I\setminus\{ j\}$ and $m^j_1\neq \widetilde{\mathbf{p}}$. Suppose that $g(\mbf m)=f(\wtd{\mbf p})$. Again we use the monotonicity axiom. Consider $q\in\Delta(S)$ with $f(\widetilde{\mathbf{p}})\mathrel{R(\widetilde{p}_{i})} q$ for all $i\in I$. Since $f$ is surjective, there is $\mathbf{p}'$ with $q=f(\mathbf{p}')$. Since  $f(\widetilde{\mathbf{p}})\mathrel{R(\widetilde{p}_{j})} q$, $j$ could make the rule choose $q$ and thus we must have $f(\widetilde{\mathbf{p}})\mathrel{R(p_{j})} q$. 

Consider $i\neq j$. Note that $g(\mbf m)=f(\wtd{\mbf p})$ implies that $f(m^j_1)\mathrel{R(\wtd{p}^j)} f(\wtd{\mbf p})$; and then $f(\wtd{\mbf p}) \mathrel{R(\wtd{p^j})} q$ implies that $q=f(\mbf p')\neq f(m^j_1)$, and therefore $\mbf p'\neq m^j_1$. This means that $i$ could make the rule choose $q$ by reporting $\mbf p'\neq m^j_1,\wtd{\mbf p}$ and winning the integer game. Again, then,  $f(\widetilde{\mathbf{p}})\mathrel{R(p_{i})} q$. By Monotonicity, we conclude that $f(\mathbf{p})=f(\widetilde{\mathbf{p}})$.

Case 3: Suppose there is $\widetilde{\mathbf{p}}\in \Delta(S)^{I}$ such that $m^i_1=\wtd{\mathbf{p}}$ for all $i\in I\setminus\{ j\}$ and $m^j_1\neq \widetilde{\mathbf{p}}$. Suppose that $g(\mbf m)=f(m^j_1)$. 

Consider $q\in\Delta(S)$. Since $f$ is surjective, there is $\mathbf{p}'$ with $q=f(\mathbf{p}')$.

Consider any $i\neq j$. Then, $i$ could make the rule choose $q$ by reporting $\mbf p'\neq m^j_1,\wtd{\mbf p}$ and winning the integer game. By the equilibrium condition it holds $f(m^j_1)\mathrel{R(p_{i})} q$. By continuity of preference, this implies $f(m^j_1)\mathrel{R(p_{i})} q$ for all $q \in \Delta(S)$. 
Thus, $f(m^j_1)$ is the maximal element for $R(p_{i})$ for all $i \ne j$. 

On the other hand, $j$ can deviate only by choosing 
$q \in L(f(\widetilde{\mathbf{p}}),\widetilde{p}_{i})$. Because of the equilibrium condition, 
$f(m^j_1)$ is the maximal element for $R(p_{j})$ in $L(f(\widetilde{\mathbf{p}}),\widetilde{p}_{i})$

Hence, by No Veto Power we conclude $f(\mathbf{p})=f(m^j_1)=g(\mathbf{m})$.

Case 4: Otherwise, everybody can deviate by announcing 
$\overline{m}^{i}=(\overline{\mathbf{p}}^{i},\overline{k}^{i})$ such that 
$f(\overline{\mathbf{p}}^{i})=p_{i}$ and $\overline{k}^{i}$ is the largest number. Since it is not profitable, it must be that 
$g(\mathbf{m})=p_{i}$. Since it is true for all $i \in I$ we reach a contradiction or unanimity.

\end{proof}

\begin{corollary}
Assume $|I|\ge 3$.  Then a belief aggregation rule 
$f$ is Nash-implementable 
if it satisfies No Veto Power and Recursive Invariance. 
\end{corollary}

\bigskip

Here are examples of belief aggregation rules which satisfy Recursive Invariance.
\begin{proposition}Two anonymous, single-valued solutions satisfying Recursive Invariance are
\begin{enumerate}
\item The geometric median for $|I|$ odd (that is, $f(\mathbf{p}) = \arg\min_{p}\sum_{i\in I}\|p_i-p\|$)
\item The equal-wealth parimutuel prices (that is, $f(\mathbf{p})$ are the equilibrium prices for the private goods linear economy where wealth is given by $1/n$)
\end{enumerate}
\end{proposition}

\begin{proof}
That the geometric median satisfies these properties is classical. Its single-valuedness is due to \cite{haldane1948note} and the monotonicity property to \cite{gini1929di}.  See also \citet{brady2015spatial,brady2017spatial}.

For the equal-wealth parimutuel prices, single-valuedness is established in \citet{eisenberggale59,eisenberg1961aggregation}.  For the monotonicity property, recall that $\rho$ is an equilibrium price if there is an allocation $x_i$ for each agent and $\alpha_i > 0$ for each agent for which:

\begin{enumerate}
\item For all $i\in I$, $\rho_s\geq \alpha_i p_{is}$
\item For all $i\in I$, $x_{is} >0$ implies $\rho_s = \alpha_i p_{is}$.
\item For all $i\in I$, $\sum_s \rho_s x_{is}=\frac{1}{n}$.
\end{enumerate}

Now, fix $\lambda\in [0,1]^I$ and take $q_i = \lambda p_i + (1-\lambda)\rho$, where $\rho$ are the parimutuel prices.  Observe then that 
$q_{is} \leq [\lambda_i/\alpha_i  + (1-\lambda_i)]\rho_s$, and that  
$q_{is} = [\lambda_i/\alpha_i  + (1-\lambda_i)]\rho_s$ when $x_{is}>0$. So by replacing $\alpha_i$ with $\alpha_i'=\frac{\alpha_i}{\lambda_i + (1-\lambda_i)\alpha_i}$, all the conditions in the preceding are satisfied with the same $x_i$ and $\rho$. \end{proof}

\bigskip

However, the the equal-wealth parimutuel prices do not satisfy No Veto Power while the geometric median does.
\begin{example}
Suppose that $|S|=2$ and assume $p_{i}=p$ for all $i=1,\cdots,n-1$ and $0<p<p_{n}<\frac{1}{n}$. Then the aggregate belief given by the parimutuel equilibrium is
\[\min\left\{p,\cdots,p,p_{n},\frac{1}{n},\cdots,\frac{n-1}{n}\right\}=p_{n},\]
which violates No Veto Power.
\end{example}

The fact that the geometric median satisfies No Veto Power for $|I|\geq 3$ follows as
\begin{eqnarray*}
    \sum_{j\in I}\|q-p_{j}\| &=& (n-1)\|q-p\| +\|q-p_{i}\| \\
    &\ge & (n-2)\|q-p\|+\|p-p_{i}\| \\
    &> & \|p-p_{i}\| \\
    &=& \sum_{j\in I}\|p-p_{j}\| 
\end{eqnarray*}
when everybody but $i$ has belief $p$ and $i$ has belief $p_{i}$ and arbitrary $q\neq p$ was taken, 
where the first inequality follows from the triangle inequality and the second as $q \neq p$.

\subsubsection*{Necessary and Sufficient Condition for Nash Implementability}
Since the parimutuel equilibrium prices violate No Veto Power and they are the \textit{Nash equilibrium} outcomes under the symmetric linear opinion pool and log utility nonetheless, they must satisfy a necessary and sufficient condition for Nash implementability.

Below is a translation of the condition by \cite{moore1990nash} to the current setting. Given $K \subset \Delta(S)$, let 
$M(K,p_{i})$ denote the set of maximal elements in $K$ according to $R(p_{i})$.

\begin{description}
\item \textbf{Condition $\mu$}: There is a set $B \subset \Delta(S)$ and a profile of 
mappings $(C_{i})_{i \in I}$ with
$C_{i}: \Delta(S)^{I} \rightarrow 
2^{\Delta(S)}$ and $f(\mathbf{p}) \in M(C_{i}(\mathbf{p}),p_{i})$ 
for each $\mathbf{p} \in \Delta(S)^{I}$,  which satisfy the following properties:
For all $\mathbf{p}, \widetilde{\mathbf{p}} \in \Delta(S)^{I}$: 
\begin{description}
\item (i) if $f(\mathbf{p})\in \bigcap_{i \in I}M(C_{i}(\mathbf{p}),\widetilde{p}_{i})$ then 
$f(\widetilde{\mathbf{p}})=f(\mathbf{p}).$
\item (ii) if $q\in M(C_{i}(\mathbf{p}),\widetilde{p}_{i})\cap \bigcap_{j \ne i} M(B,\widetilde{p}_{j})$ then 
$f(\widetilde{\mathbf{p}})=q.$
\item (iii) if $q\in  \bigcap_{i \in I} M(B,\widetilde{p}_{i})$ then 
$f(\widetilde{\mathbf{p}})=q.$ 
\end{description}
\end{description}

Below is a natural analogue of the result by \cite{moore1990nash}, which can be proven by replacing the mechanism in the proof of Proposition \ref{prop:Nash} by the one analogous to Moore-Repullo's.
\begin{proposition}
Assume $|I|\ge 3$. Then a belief aggregation rule is Nash-implementable if and only if it satisfies Condition $\mu$.    
\end{proposition}

We can apply the condition to the parimutuel equilibrium prices under the symmetric linear opinion pool and log utility by letting $B=\Delta(S)$ and
\[C_{i}(\mathbf{p})=\left\{\frac{1}{n}\left(\rho_{i}+\sum_{j\ne i}\rho_{j}(\mathbf{p})\right): \rho_{i} \in \Delta(S)\right\},\]
where $\rho_{j}(\mathbf{p})$ denotes the equilibrium belief reporting by $j$ when the belief profile is $\mathbf{p}$.

To see the that Condition (ii) is strictly weaker than No Veto Power, consider the two-state case in which all but individual $i$ have the same belief $p$. Then the presumption in Condition (ii) is met only when $\widetilde{p}_{i}=p$ or $\widetilde{p}_{i}\le p\le \frac{1}{n}$ or $\frac{n-1}{n}\le p \le \widetilde{p}_{i}$, and otherwise the conclusion $f(\widetilde{p}_{i},p,\cdots,p)=p$ need not hold.

\section{Concluding Remarks}
We conclude by discussing future directions.

A natural question is if we can have any nice strategy-proof rule other than dictatorship or a constant rule.
\begin{description}
\item \textbf{Strategy-proofness}: For all $\mathbf{p} \in \Delta(S)^{I}$, $i \in I$ and $q_{i} \in \Delta(S)$, it holds
\[f(\mathbf{p})R(p_{i})f(q_{i},\mathbf{p}_{-i}) .\]
\end{description}

We can verify that Monotonicity is necessary for 
strategy-proofness.
\begin{lemma}
Strategy-proofness implies Monotonicity.
\end{lemma}
\begin{proof}
Suppose $\mathbf{p}, \widetilde{\mathbf{p}} \in \Delta(S)^{I}$ meet the presumption of Monotonicity, and 
suppose $f(\widetilde{p}_{i}, \mathbf{p}_{-i}) \ne f(\mathbf{p})$. Then there are two cases.

Case 1: Suppose $f(\widetilde{p}_{i}, \mathbf{p}_{-i}) \in L(f(\mathbf{p}),p_{i})$, then because 
$L(f(\mathbf{p}),p_{i}) \cap 
L(f(\mathbf{p}),\widetilde{p}_{i})=\{f(\mathbf{p})\}$ holds in this domain it must be that 
\[f(\mathbf{p}) P(\widetilde{p}_{i}) f(\widetilde{p}_{i}, \mathbf{p}_{-i}) ,\]
which is a violation of Strategy-proofness at $(\widetilde{p}_{i}, \mathbf{p}_{-i})$.

Case 2: Suppose $f(\widetilde{p}_{i}, \mathbf{p}_{-i}) \notin L(f(\mathbf{p}),p_{i})$, 
then it means
\[  f(\widetilde{p}_{i}, \mathbf{p}_{-i}) P(p_{i}) f(\mathbf{p}),\]
which is a violation of Strategy-proofness at $\mathbf{p}$.

Hence $f(\widetilde{p}_{i}, \mathbf{p}_{-i}) = f(\mathbf{p})$.
By repeated application of this argument we obtain 
$f(\widetilde{\mathbf{p}}) = f(\mathbf{p})$
\end{proof}

\bigskip

As Monotonicity is equivalent to Recursive Invariance, we can say that if a rule does not satisfy Recursive Invariance it fails to be strategy-proof. A sufficient condition for strategy-proofness remains to be an open question, though. The current domain does not satisfy the richness condition (R1) by \cite{klaus2013relation}, which guarantees sufficiency of Monotonicity for strategy-proofness.

Since we choose a single point from a convex subset of an Euclidian space, our problem falls in the model of choosing a pure public good as considered by  \cite{zhou1991impossibility}, who showed that any strategy-proof rule can have no more than one-dimensional range. In our context, this means we have possibility only in the two-state case in which the median rule works. The current domain is strictly smaller than the one considered by Zhou, however, since we assume that everybody has a fixed and common vNM index so that specifying a belief determines a whole preference uniquely, while Zhou allows all convex preferences or at least needs to allow all quadratic preferences. Since preferences cannot be separable over the budget set which is not a product set, the possibility results based on separable preferences (such as \cite{border1983straightforward}) do not apply.  Thus, possibility/impossibility remains to be an open question.

\bibliographystyle{ecta}
\bibliography{references}

\end{document}